\documentclass[a4paper]{article}

\usepackage{microtype}

\usepackage{fullpage}

\usepackage[boxed]{algorithm2e}

\usepackage{amssymb,amsmath,amsthm,graphics,graphicx,enumerate}

\newcommand{\NP}{{\sf NP}}
\newcommand{\NPc}{{\sf NPc}}

\newcommand{\W}{{\sf W}}
\renewcommand{\P}{{\sf P}}

\newcommand{\pname}{\textsc}

\newcommand{\algMinSup}{\textsf{\textup{MinSup}}}
\newcommand{\algMaxInf}{\textsf{\textup{MaxInf}}}
\newcommand{\algWInt}{\textsf{\textup{wCOC}}}
\newcommand{\ucoc}{\textsf{\textup{uCOC}}}
\newcommand{\otilde}{O} 
\newcommand{\vi}{\pname{Vertex Integrity}}
\newcommand{\wvi}{\pname{Weighted Vertex Integrity}}
\newcommand{\coc}{\pname{Component Order Connectivity}}
\newcommand{\wcoc}{\pname{Weighted Component Order Connectivity}}
\newcommand{\wcc}{w_{\text{cc}}}

\DeclareMathOperator{\dpt}{dp}

\newtheorem{theorem}{Theorem}
\newtheorem{lemma}{Lemma}[section]

\theoremstyle{definition}

\theoremstyle{remark}

\usepackage{tikz} 
\usetikzlibrary{matrix}
\usetikzlibrary{arrows,shapes,positioning}
\usetikzlibrary{decorations.pathreplacing}

\usepackage{hyperref}
\hypersetup{%
  colorlinks=true,
  linkcolor=blue,%
  citecolor=blue
}

\title{On the Computational Complexity of Vertex Integrity\\%
  and Component Order Connectivity%
  \footnote{A preliminary version of this paper already appeared in
    the conference proceedings of
    ISAAC~2014~\cite{drange14onthe}.\newline%
    The research leading to these results has received funding from
    the Research Council of Norway, Bergen Research Foundation under
    the project Beating Hardness by Preprocessing and the European
    Research Council under the European Union's Seventh Framework
    Programme (FP/2007-2013) / ERC Grant Agreement n.~267959.
  }%
}
\author{P\aa l Gr\o n\aa s Drange \and Markus Sortland Dregi \and Pim
  van 't Hof}

\begin{document}

\maketitle

\begin{abstract}
  The \pname{Weighted Vertex Integrity} (wVI) problem takes as input
  an $n$-vertex graph $G$, a weight function $w:V(G)\rightarrow
  \mathbb{N}$, and an integer $p$.  The task is to decide if there
  exists a set $X\subseteq V(G)$ such that the weight of $X$ plus the
  weight of a heaviest component of $G-X$ is at most~$p$.  Among other
  results, we prove that:
\begin{enumerate}[(1)]
\item wVI is \NP-complete on co-comparability graphs, even if 
each vertex has weight~$1$;
\item wVI can be solved in $O(p^{p+1}n)$ time;
\item wVI admits a kernel with at most $p^3$ vertices.
\end{enumerate}
Result (1) refutes a conjecture by Ray and Deogun
(J.~Combin.~Math.~Combin.~Comput.\ 16: 65--73, 1994) and answers an
open question by Ray et al.~(Ars Comb.~79: 77--95, 2006).  It also
complements a result by Kratsch et al.~(Discr.~Appl.~Math.~77:
259--270, 1997), stating that the unweighted version of the problem
can be solved in polynomial time on co-comparability graphs of bounded
dimension, provided that an intersection model of the input graph is
given as part of the input.

An instance of the \pname{Weighted Component Order Connectivity}
(wCOC) problem consists of an $n$-vertex graph $G$, a weight function
$w:V(G)\rightarrow \mathbb{N}$, and two integers $k$ and $\ell$, and
the task is to decide if there exists a set $X\subseteq V(G)$ such
that the weight of $X$ is at most~$k$ and the weight of a heaviest
component of $G-X$ is at most~$\ell$.  In some sense, the wCOC problem
can be seen as a refined version of the wVI problem.  We obtain
several classical and parameterized complexity results on the wCOC
problem, uncovering interesting similarities and differences between
wCOC and wVI.  We prove, among other results, that:
\begin{enumerate}[(4)]
\item[(4)] wCOC can be solved in $O(\min\{k,\ell\}\cdot n^3)$ time on
  interval graphs, while the unweighted version can be solved in
  $O(n^2)$ time on this graph class;
\item[(5)] wCOC is \W[1]-hard on split graphs when parameterized by $k$ or by $\ell$;
\item[(6)] wCOC can be solved in $2^{O(k\log \ell)} n$ time;
\item[(7)] wCOC admits a kernel with at most $k\ell(k+\ell)+k$
  vertices.
\end{enumerate}
We also show that result (6) is essentially tight by proving that wCOC
cannot be solved in $2^{o(k \log \ell)}n^{O(1)}$ time, even when
restricted to split graphs, unless the Exponential Time Hypothesis
fails.
\end{abstract}

\section{Introduction}
\label{s-intro}

Motivated by a multitude of practical applications, many different
vulnerability measures of graphs have been introduced in the
literature over the past few decades.  The vertex and edge
connectivity of a graph, although undoubtedly being the most
well-studied of these measures, often fail to capture the more subtle
vulnerability properties of networks that one might wish to consider,
such as the number of resulting components, the size of the largest or
smallest component that remains, and the largest difference in size
between any two remaining components.  The two vulnerability measures
we study in this paper, \emph{vertex integrity} and \emph{component
  order connectivity}, take into account not only the number of
vertices that need to be deleted in order to break a graph into
pieces, but also the number of vertices in the largest component that
remains.

The \emph{vertex integrity} of an unweighted graph $G$ is defined as
$\iota(G) = \min\{|X|+n(G-X) \mid X\subseteq V(G)\}$, where $n(G-X)$
is the number of vertices in the largest connected component of $G-X$.
This vulnerability measure was introduced by Barefoot, Entringer, and
Swart~\cite{BarefootES87} in 1987.  For an overview of structural
results on vertex integrity, including combinatorial bounds and
relationships between vertex integrity and other vulnerability
measures, we refer the reader to a survey on the subject by Bagga et
al.~\cite{bagga1992survey}.  We mention here only known results on the
computational complexity of determining the vertex integrity of a
graph.

The \pname{Vertex Integrity} (VI) problem takes as input an
$n$-vertex graph $G$ and an integer~$p$, and asks whether
$\iota(G)\leq p$.  This problem was shown to be \NP-complete, even
when restricted to planar graphs, by Clark, Entringer, and
Fellows~\cite{ClarkEF87}.  On the positive side, Fellows and
Stueckle~\cite{fellows1989immersion} showed that the problem can be
solved in $O(p^{3p} n)$ time, and is thus fixed-parameter tractable
when parameterized by~$p$.  In the aforementioned survey, Bagga et
al.~\cite{bagga1992survey} mention that \vi{} can be solved in
$O(n^3)$ time when the input graph is a tree or a cactus graph.
Kratsch, Kloks, and M\"{u}ller~\cite{KratschKM97} studied the
computational complexity of determining the value of several
vulnerability measures in classes of intersection graphs.  Their
results imply that \vi{} can be solved in $O(n^3)$ time on interval
graphs, in $O(n^4)$ time on circular-arc graphs, and in $O(n^5)$ time
on permutation graphs and trapezoid graphs.  Kratsch et
al.~\cite{KratschKM97} also mention that the problem can be solved in
$O(n^{2d+1})$ time on co-comparability graphs of dimension at
most~$d$, provided that an intersection model of the input graph is
given as part of the input.

Ray and Deogun~\cite{RayD94} were the first to study the more general
\wvi{} (wVI) problem.  This problem takes as input an $n$-vertex graph
$G$, a weight function $w:V(G)\rightarrow \mathbb{N}$, and an integer
$p$.  The task is to decide if there exists a set $X\subseteq V(G)$
such that the weight of $X$ plus the weight of a heaviest component of
$G-X$ is at most~$p$.  Using a reduction from 0-1 \pname{Knapsack},
Ray and Deogun~\cite{RayD94} identified several graph classes on which
the \pname{Weighted Vertex Integrity} problem is weakly \NP-complete.
In particular, their result implies that the problem is weakly
\NP-complete on trees, bipartite graphs, series-parallel graphs, and
regular graphs, and therefore also on superclasses such as chordal
graphs and comparability graphs.  A common property of these classes
is that they contain graphs with arbitrarily many asteroidal triples
and induced paths on five vertices; any graph class that does not have
this property is not covered by the result of Ray and Deogun.  They
conjectured that the \wvi{} problem can be solved in polynomial time
on co-comparability graphs, a well-known example of a class of graphs
that do not contain asteroidal triples at all.  More than a decade
later, Ray et al.~\cite{ray2006weighted} presented a polynomial-time
algorithm for \pname{Weighted Vertex Integrity} on interval graphs, a
subclass of co-comparability graphs.  In the same paper, they pointed
out that the complexity of the problem on co-comparability graphs
remained unknown.

We now turn our attention to the second vulnerability measure studied
in this paper.  For any positive integer $\ell$, the
\emph{$\ell$-component order connectivity} of a graph $G$ is defined
to be the cardinality of a smallest set $X\subseteq V(G)$ such that
$n(G-X)<\ell$.  We refer to the survey by Gross et al.~\cite{Gross+13}
for more background on this graph parameter.  Motivated by the
definitions of $\ell$-component order connectivity and the
\pname{Weighted Vertex Integrity} problem, we introduce the
\pname{Weighted Component Order Connectivity} (wCOC) problem.  This
problem takes as input a graph $G$, a weight function
$w:V(G)\rightarrow \mathbb{N}$, and two integers $k$ and $\ell$.  The
task is to decide if there exists a set $X\subseteq V(G)$ such that
the weight of $X$ is at most~$k$ and the weight of a heaviest
component of $G-X$ is at most~$\ell$.  Observe that the \wcoc{}
problem can be interpreted as a more refined version of \wvi.  We
therefore find it surprising that, to the best of our knowledge, the
\wcoc{} problem has not yet been studied in the literature.  We do
however point out that the techniques described by Kratsch et
al.~\cite{KratschKM97} yield polynomial-time algorithms for the
unweighted version of the problem on interval graphs, circular-arc
graphs, permutation graphs, and trapezoid graphs, and that very
similar problems have received some attention
recently~\cite{Ben-AmeurMN13,Gross+13}.

\medskip
\noindent {\bf Our Contribution.} In Section~\ref{s-vi}, we present
our results on \vi{} and \wvi{}.  We show that VI is \NP-complete on
co-bipartite graphs, and hence on co-comparability graphs.  This
refutes the aforementioned conjecture by Ray and Deogun~\cite{RayD94}
and answers an open question by Ray et al.~\cite{ray2006weighted}.  It
also forms an interesting contrast with the result by Kratsch et
al.~\cite{KratschKM97} stating that VI can be solved in $O(n^{2d+1})$
time on co-comparability graphs of dimension at most~$d$ if an
intersection model is given as part of the input.  We also show that
even though VI can be solved in linear time on split graphs, the
problem remains \NP-complete on chordal graphs.  Interestingly, we
prove that unlike the unweighted variant of the problem, the wVI
problem is \NP-complete when restricted to split graphs; observe that
this does not follow from the aforementioned hardness result by Ray
and Deogun~\cite{RayD94}, as split graphs do not contain induced paths
on five vertices.

Recall that Fellows and Stueckle~\cite{fellows1989immersion} showed
that VI can be solved in $O(p^{3p}n)$ time on general graphs.  We
strengthen this result by showing that even the wVI problem can be
solved in $O(p^{p+1} n)$ time.  We also show that wVI admits a kernel
with at most $p^3$ vertices, each having weight at most~$p$.

\setlength{\tabcolsep}{10pt} 
\renewcommand{\arraystretch}{1.1} 
\begin{table}[htb]
  \centering
  \begin{tabular}{| l | c | c | c | c |}
    \hline
    & VI & wVI & COC & wCOC \\
    \hline
    general & \NPc~\cite{ClarkEF87} & \NPc~\cite{ClarkEF87} & \NPc~\cite{ClarkEF87} & \NPc~\cite{ClarkEF87}\\
    co-bipartite & \NPc & \NPc & \NPc & \NPc \\
    chordal & \NPc & \NPc & \NPc & \NPc \\
    split     & $O(n+m)$~\cite{LiZZ08}  & \NPc  & \NPc & \NPc \\
    interval & $O(n^3)$~\cite{KratschKM97} & $O(n^6\log
    n)$~\cite{ray2006weighted}  & $O(n^2)$ & $O(\min\{k,\ell\}\cdot
    n^3)$ \\
    complete & $O(n)$ & $O(n)$ & $O(n)$ & weakly \NPc\\
    \hline
  \end{tabular}
  \vspace{.2cm}
  \caption{An overview of the classical complexity results proved in this
    paper.   Previously known results are given with a reference.}
  \label{t-results}
\end{table}

Section~\ref{s-coc} contains our results on \coc{} and \wcoc{}.  The
observation that there is a polynomial-time Turing reduction from VI
to COC implies that
the latter problem cannot be solved in polynomial time on any graph
class for which VI is \NP-complete, unless \P=\NP.
We prove that wCOC is weakly \NP-complete already on complete graphs,
while the unweighted variant of the problem, which is trivial on
complete graphs, remains \NP-complete when restricted to split graphs.
We find the latter result particularly interesting in light of
existing polynomial-time algorithms for computing similar (unweighted)
vulnerability measures of split graphs, such as
toughness~\cite{Woeginger98}, vertex integrity, scattering number,
tenacity, and rupture degree~\cite{LiZZ08}.  To complement our
hardness results, we present a pseudo-polynomial-time algorithm that
solves the wCOC problem in $O(\min\{k,\ell\}\cdot n^3)$ time on
interval graphs.  We then modify this algorithm to solve the
unweighted version of the problem in $O(n^2)$ time on interval graphs,
thereby improving the $O(n^3)$-time algorithm that follows from the
results by Kratsch et al.~\cite{KratschKM97}.  Observe that the
aforementioned hardness results rule out the possibility of solving
wCOC in polynomial time on interval graphs or in pseudo-polynomial
time on split graphs, unless \P{} = \NP.

In Section~\ref{s-coc}, we also completely classify the parameterized
and kernelization complexity of COC and wCOC on general graphs with
respect to the parameters $k$, $\ell$, and $k+\ell$.  We first observe
that both problems are para-\NP-hard when parameterized by~$\ell$ due
to the fact that COC is equivalent to \pname{Vertex Cover} when $\ell
= 1$.  We then prove that if we take either $k$ or $\ell$ to be the
parameter, then COC is $\W[1]$-hard even on split graphs.  On the
positive side, we show that wCOC becomes fixed-parameter tractable
when parameterized by $k+\ell$.  We present an algorithm for solving
the problem in time $2^{O(k\log \ell)} n$ time, before proving that
the problem cannot be solved in time $2^{o(k\log \ell)} n^{O(1)}$
unless the Exponential Time Hypothesis fails.  Finally, we show that
wCOC admits a polynomial kernel with at most $k\ell(k+\ell)+k$
vertices, where each vertex has weight at most~$k+\ell$.

\section{Preliminaries}

All graphs considered in this paper are finite, undirected, and
simple.  We refer to the monograph by Diestel~\cite{Diestel05} for
graph terminology and notation not defined here.  For more information
on parameterized complexity and kernelization, we refer to the book by
Downey and Fellows~\cite{ParameterizedComplexity}.  For definitions
and characterizations of the graph classes mentioned in this paper, as
well as the inclusion relationships between those classes, we refer to
the survey by Brandst\"{a}dt, Le, and Spinrad~\cite{BrandstadtLS99}.
Whenever we write that a (weighted) problem is \NP-complete, we mean
strongly \NP-complete, unless specifically stated otherwise.

Let $G$ be a graph and $w:V(G)\rightarrow \mathbb{N} = \{0,1,\ldots\}$ a weight
function on the vertices of~$G$.  The weight of a subset $X\subseteq
V(G)$ is defined as $w(X) = \sum_{v\in X} w(v)$.
We define $\wcc(G)$ to be the weight of a heaviest component of $G$,
i.e., $\wcc(G) = \max\{ w(V(G_i)) \mid 1\leq i\leq r\}$, where
$G_1,\ldots,G_r$ are the components of $G$.  The \emph{weighted vertex
  integrity} of $G$ is defined as
\[ 
\iota(G) = \min\{w(X)+\wcc(G-X) \mid X\subseteq V(G)\} \, ,
\]
where $G-X$ denotes the graph obtained from $G$ by deleting all the
vertices in $X$.  Any set $X\subseteq V(G)$ for which $w(X)+\wcc(G-X)
= \iota(G)$ is called an \emph{$\iota$-set} of~$G$.  We consider the
following two decision problems:

\medskip
\indent \pname{Weighted Vertex Integrity} (wVI)\\
\indent \emph{Instance:} \hspace*{-.01cm} A graph $G$, a weight function
$w:V(G)\rightarrow \mathbb{N}$, and an integer~$p$.\\
\indent \emph{Question:} Is $\iota(G) \leq p$?

\medskip
\indent \pname{Weighted Component Order Connectivity} (wCOC)\\
\indent \emph{Instance:} \hspace*{-.01cm} A graph $G$, a weight function
$w:V(G)\rightarrow \mathbb{N}$, and two integers~$k$ and~$\ell$.\\
\indent \emph{Question:} Is there a set $X\subseteq V(G)$ with $w(X)\leq k$ such
that $\wcc(G-X) \leq \ell$?

\medskip
\noindent
The unweighted versions of these two problems, where $w(v) = 1$ for
every vertex $v\in V(G)$, are called \pname{Vertex Integrity} (VI)
and \pname{Component Order Connectivity} (COC), respectively.

A \emph{bipartite graph} is a graph whose vertex set can be
partitioned into two independent sets, and a graph $G$ is
\emph{co-bipartite} if its complement $\overline{G}$ is bipartite.  A
\emph{split graph} is a graph whose vertex set can be partitioned into
a clique $C$ and an independent set $I$; such a partition $(C,I)$ is
called a \emph{split partition}.  A split graph $G$ with split
partition $(C,I)$ and edge set $E$ is denoted by $G = (C,I,E)$.  Note
that, in general, a split graph can have more than one split
partition.
A graph is \emph{chordal} if it has no induced cycle of length more
than~$3$.  Let ${\cal F}$ be a family of non-empty sets.  The
\emph{intersection graph} of ${\cal F}$ is obtained by representing
each set in ${\cal F}$ by a vertex and making two vertices adjacent if
and only if their corresponding sets intersect.  A graph is an
\emph{interval graph} if it is the intersection graph of intervals on
the real line.  A \emph{comparability graph} is a graph that admits a
transitive orientation, that is, a graph whose edges can be directed
in such a way that whenever $(u,v)$ and $(v,w)$ are directed edges,
then so is $(u,w)$.  A graph $G$ is a \emph{co-comparability graph} if
its complement $\overline{G}$ is a comparability graph.

It is well know that split graphs and interval graphs form two
incomparable subclasses of chordal graphs~\cite{BrandstadtLS99}.
Every bipartite graph is a comparability graph, as directing all the
edges of a bipartite graph from one bipartition class to the other
yields a transitive orientation.  Consequently, co-bipartite graphs
form a subclass of co-comparability graphs.  Interval graphs form
another subclass of co-comparability graphs; this readily follows from
the fact that co-comparability graphs are exactly the intersection
graphs of continuous real-valued functions over some interval
$I$~\cite{GolumbicRU83}.

A \emph{parameterized problem} is a subset $Q\subseteq \Sigma^*\times \mathbb{N}$ for
some finite alphabet $\Sigma$, where the second part of the input is
called the \emph{parameter}.  A parameterized problem $Q\subseteq
\Sigma^*\times \mathbb{N}$ is said to be \emph{fixed-parameter
  tractable} if for each pair $(x,k)\in \Sigma^*\times \mathbb{N}$ it
can be decided in time $f(k)\, |x|^{O(1)}$ whether $(x,k)\in Q$, for
some function $f$ that only depends on $k$; here, $|x|$ denotes the
length of input $x$.  We say that a parameterized problem $Q$ has a
\emph{kernel} if there is an algorithm that transforms each instance
$(x,k)$ in time $(|x|+k)^{O(1)}$ into an instance $(x',k')$, such that
$(x,k)\in Q$ if and only if $(x',k')\in Q$ and $|x'|+k'\leq g(k)$ for
some function $g$.
If $g$ is a polynomial, then we say that the problem has a
\emph{polynomial} kernel.

\section{Vertex Integrity}
\label{s-vi}

As mentioned in the introduction, Ray et al.~\cite{ray2006weighted}
asked whether \pname{Weighted Vertex Integrity} can be solved in
polynomial time on co-comparability graphs.  We show that this is not
the case, unless $\P = \NP$.  In fact, we prove a much stronger result
in Theorem~\ref{t-cobipartite} below by showing \NP-completeness of an
easier problem (\pname{Vertex Integrity}) on a smaller graph class
(co-bipartite graphs).

\begin{theorem}
  \label{t-cobipartite}
  \vi{} is \NP-complete on co-bipartite graphs.
\end{theorem}

\begin{proof}
  The problem is clearly in \NP.  To show that it is \NP-hard, we give
  a polynomial-time reduction from the \pname{Balanced Complete
    Bipartite Subgraph} problem.  This problem, which is known to be
  \NP-complete~\cite{garey1979computers}, takes as input a bipartite
  graph $G = (A,B,E)$ and an integer $k\geq 1$, and asks whether there
  exist subsets $A'\subseteq A$ and $B'\subseteq B$ such that $|A'| =
  |B'| = k$ and $G[A'\cup B']$ is a complete bipartite graph.  Let
  $(G,k)$ be an instance of \pname{Balanced Complete Bipartite
    Subgraph}, where $G = (A,B,E)$ is a bipartite graph on $n$
  vertices.  We claim that $(G,k)$ is a yes-instance of
  \pname{Balanced Complete Bipartite Subgraph} if and only if
  $(\overline{G},n-k)$ is a yes-instance of \vi{}.

  Suppose there exist subsets $A'\subseteq A$ and $B'\subseteq B$ such that $|A'| =
  |B'| = k$ and $A'\cup B'$ induces a complete bipartite subgraph in
  $G$.  Observe that in $\overline{G}$, both $A'$ and $B'$ are
  cliques, and there is no edge between $A'$ and $B'$.  Hence, if we
  delete all the vertices in $V(\overline{G})\setminus (A'\cup B')$
  from $\overline{G}$, the resulting graph has exactly two components
  containing exactly $k$ vertices each.  Since
  $|V(\overline{G})\setminus (A'\cup B')| = n-2k$, it holds that
  $\iota(\overline{G}) \leq n - 2k + k = n - k$, and hence
  $(\overline{G},n-k)$ is a yes-instance of \vi{}.

  For the reverse direction, suppose $(\overline{G},n-k)$ is a
  yes-instance of \vi{}.  Then there exists a subset $X \subseteq
  V(\overline{G})$ such that $|X|+n(\overline{G}-X) \leq n-k$.  The
  assumption that $k\geq 1$ implies that $\overline{G}-X$ is
  disconnected, as otherwise $|X|+n(\overline{G}-X) = V(\overline{G})
  = n$.  Let $A' = A \setminus X$ and $B' = B \setminus X$.  Since
  $\overline{G}$ is co-bipartite, both $A'$ and $B'$ are cliques.
  Moreover, since $\overline{G}-X$ is disconnected, there is no edge
  between $A'$ and $B'$.  Hence, $\overline{G}[A']$ and
  $\overline{G}[B']$ are the two components of $\overline{G}-X$.
  Without loss of generality, suppose that $|A'| \geq |B'|$.  Then
  $|B'| = n - (|X|+|A'|) = n - (|X| + n(\overline{G} - X)) \geq n -
  (n-k) = k$ and hence $|A'| \geq |B'| \geq k$.  This, together with
  the observation that $A'\cup B'$ induces a complete bipartite
  subgraph in $G$, implies that $(G,k)$ is a yes-instance of
  \pname{Balanced Complete Bipartite Subgraph}.
\end{proof}

Ray and Deogun~\cite{RayD94} proved that \pname{Weighted Vertex
  Integrity} is \NP-complete on any graph class that satisfies certain
conditions.  Without explicitly stating these (rather technical)
conditions here, let us point out that any graph class satisfying
these conditions must contain graphs with arbitrarily many asteroidal
triples and induced paths on five vertices.
Theorem~\ref{t-cobipartite} shows that neither of these two properties
is necessary to ensure \NP-completeness of \pname{Weighted Vertex
  Integrity}, since co-bipartite graphs contain neither asteroidal
triples nor induced paths on five vertices.

In Theorem~\ref{t-splithard} below, we show that \wvi{} is
\NP-complete on split graphs.  Since split graphs do not contain
induced paths on five vertices, this graph class is not covered by the
aforementioned hardness result of Ray and Deogun~\cite{RayD94}.

\begin{lemma}
  \label{l-simplicial}
  For every graph $G$ and weight function $w:V(G)\rightarrow \mathbb{N}$, there
  exists an $\iota$-set $X$ that contains no simplicial vertices of $G$.
\end{lemma}
\begin{proof}
  Let $w:V(G)\rightarrow \mathbb{N}$ be a weight function of a graph $G$, and
  let $X$ be an $\iota$-set of $G$ containing a simplicial vertex $s$.
  Observe that $s$ is adjacent to at most one component of $G-X$.  Let
  $X' = X\setminus \{s\}$.  We claim that $X'$ is an $\iota$-set of $G$.

  Let $G_1,\ldots,G_r$ denote the components of $G-X'$, and without
  loss of generality assume that $s\in V(G_1)$.  The fact that $s$ is
  a simplicial vertex of $G$ implies that $G_2,\ldots,G_r$ are
  components of $G-X$ as well.  Hence $\wcc(G_i) \leq \wcc(G-X)$ for
  every $i\in \{2,\ldots,r\}$.  Let us determine an upper bound on
  $\wcc(G_1)$.  If $s$ is adjacent to a component $H$ in the graph
  $G-X$, then $\wcc(G_1) = \wcc(H) + w(s) \leq \wcc(G-X)+w(s)$.
  Otherwise, $s$ is an isolated vertex in $G-X'$, implying that
  $\wcc(G_1) = w(s)$.  We find that $\wcc(G-X') = \max\{\wcc(G_i) \mid
  1\leq i\leq r\}\leq \wcc(G-X)+w(s)$.  Consequently,
  $w(X')+\wcc(G-X')\leq (w(X)-w(s))+(\wcc(G-X)+w(s)) = w(X)+\wcc(G-X)
  = \iota(G)$, where the last equality follows from the assumption
  that $X$ is an $\iota$-set of $G$.  We conclude that $X'$ is an
  $\iota$-set of~$G$.
\end{proof}

Given a graph~$G$, the \emph{incidence split graph} of~$G$ is the
split graph $G^* = (C^*,I^*,E^*)$ whose vertex set consists of a
clique $C^* = \{v_x \mid x\in V(G)\}$ and an independent set $I^* =
\{v_e \mid e\in E(G)\}$, and where two vertices $v_x\in C^*$ and
$v_e\in I^*$ are adjacent if and only if the vertex~$x$ is incident
with the edge~$e$ in~$G$.  The following lemma will be used in the
proofs of hardness results not only in this section, but also in
Section~\ref{s-coc}.

\begin{lemma}
\label{l-incidence}
Let $G = (V,E)$ be a graph, $G^* = (C^*,I^*,E^*)$ its incidence split
graph, and $k<|V|$ a non-negative integer.  Then the following
statements are equivalent:
\begin{enumerate}[(i)]
\item $G$ has a clique of size~$k$;
\item there exists a set $X\subseteq C^*$ such that $|X| \leq k$ and
  $|X|+n(G^*-X)\leq |V|+|E|-\binom{k}{2}$;
\item there exists a set $X\subseteq C^*$ such that $|X|\leq k$ and $n(G^*-X)\leq
  |V|+|E|-\binom{k}{2}-k$.
\end{enumerate}
\end{lemma}
\begin{proof}
  Let $n = |V|$ and $m = |E|$.  We first prove that (i) implies (iii).
  Suppose $G$ has a clique $S$ of size~$k$.  Let $X = \{v_x\in C^*
  \mid x\in S\}$ denote the set of vertices in $G^*$ corresponding to
  the vertices of $S$.  Similarly, let $Y = \{v_e\in I^* \mid e\in
  E(G[S])\}$ denote the set of vertices in $G^*$ corresponding to the
  edges in $G$ both endpoints of which belong to $S$.  Observe that
  $|Y| = \binom{k}{2}$ due to the fact that $S$ is a clique of size
  $k$ in $G$.  Now consider the graph $G^* - X$.  In this graph, every
  vertex of $Y$ is an isolated vertex, while every vertex of
  $I^*\setminus Y$ has at least one neighbor in the clique
  $C^*\setminus X$.  This implies that $n(G^* - X) = n + m -
  \binom{k}{2}-k$.

  Since (iii) trivially implies (ii), it remains to show that (ii)
  implies (i).  Suppose there exists a set $X\subseteq C^*$ such that
  $|X| \leq k$ and $|X|+n(G^*-X)\leq |V|+|E|-\binom{k}{2}$.  Let
  $Z\subseteq I^*$ be the set of vertices in $I^*$ both neighbors of
  which belong to $X$.  Observe that $|Z|\leq \binom{|X|}{2}$ and
  $n(G^*-X) = n+m-|X|-|Z|$.  Hence
  \[
  n+m-\binom{k}{2} \geq |X|+n(G^*-X) = n+m-|Z| \geq n+m-\binom{|X|}{2} \, ,
  \]
  which implies that $\binom{k}{2}\leq \binom{|X|}{2}$.  Since $|X|\leq
  k$ by assumption, we find that $|X| = k$ and all the above
  inequalities must be equalities.  In particular, we find that $|Z| =
  \binom{|X|}{2}$.  We conclude that the vertices in~$G$ that
  correspond to~$X$ form a clique of size~$k$ in~$G$.
\end{proof}

\begin{theorem}
\label{t-splithard}
\pname{Weighted Vertex Integrity} is \NP-complete on split graphs.
\end{theorem}

\begin{proof}
  We give a reduction from the \NP-hard problem \pname{Clique}.  Given
  an instance $(G,k)$ of \pname{Clique} with $n = |V(G)|$ and $m =
  |E(G)|$, we create an instance $(G',w,p)$ of \pname{Weighted Vertex
    Integrity} as follows.  To construct $G'$, we start with the
  incidence split graph $G^* = (C^*,I^*,E^*)$ of $G$, and we add a
  single isolated vertex $z$.  We define the weight function $w$ by
  setting $w(z) = n+m-\binom{k}{2}-k$ and $w(v) = 1$ for every $v\in
  V(G')\setminus \{z\}$.  Finally, we set $p = n+m-\binom{k}{2}$.  For
  convenience, we assume that $k<n$ and that $\binom{k}{2} \leq m$.
  We now claim that $G$ has a clique of size $k$ if and only if
  $\iota(G')\leq p$.  Since $G'$ is a split graph and all the vertex
  weights are polynomial in~$n$, this suffices to prove the theorem.

  
  First suppose~$G$ has a clique~$S$ of size~$k$.  By
  Lemma~\ref{l-incidence}, there exists a set~$X \subseteq C^*$ such
  that $|X|\leq k$ and $n(G^*-X)\leq n+m-\binom{k}{2}-k$.  Since $w(z)
  = n+m-\binom{k}{2}-k$ and every other vertex in~$G'$ has weight~$1$,
  it follows that $\wcc(G'-X) = n+m-\binom{k}{2}-k$.  Consequently,
  $w(X)+\wcc(G'-X)\leq n+m-\binom{k}{2} = p$, so we conclude that
  $\iota(G')\leq p$.


  For the reverse direction, suppose $\iota(G')\leq p$, and let $X\subseteq V(G')$ be
  an~$\iota$-set of~$G'$.  Due to Lemma~\ref{l-simplicial}, we may
  assume that $X \subseteq C^*$.  We claim that $|X|\leq k$.  For
  contradiction, suppose $|X|\geq k+1$.  Then $w(X) = |X|\geq k+1$ and
  $\wcc(G'-X)\geq w(z) = p-k$.  This implies that $\iota(G') =
  w(X)+\wcc(G'-X)\geq p+1$, yielding the desired contradiction.  Now
  let~$H$ be the component of~$G'-X$ containing the clique $C^*
  \setminus X$.  The fact that every vertex in $V(G')\setminus \{z\}$
  has weight~$1$ and the assumption that $k<n$ imply that $|V(H)| =
  n(G^*-X) = \wcc(G^*-X)$.  Now observe that $|X|+n(G^*-X)\leq
  |X|+\max\{w(Z),\wcc(G^*-X)\} = \iota(G')\leq p = n+m-\binom{k}{2}$.
  We can therefore invoke Lemma~\ref{l-incidence} to conclude that~$G$
  has a clique of size~$k$.
\end{proof}

The following result, previously obtained by Li et al.~\cite{LiZZ08},
is an easy consequence of Lemma~\ref{l-simplicial}.
Theorem~\ref{t-chordal} below shows that this result is in some sense
best possible.

\begin{theorem}[\cite{LiZZ08}]
  \label{t-spliteasy}
  \pname{Vertex Integrity} can be solved in linear time on split
  graphs.
\end{theorem}

\begin{theorem}
  \label{t-chordal}
  \pname{Vertex Integrity} is \NP-complete on chordal graphs.
\end{theorem}
\begin{proof}
  We describe a slight modification of the reduction in the proof of
  Theorem~\ref{t-splithard}.  Given an instance $(G,k)$ of
  \pname{Clique}, we construct a graph $G''$ in the same way as we
  constructed the graph $G'$, but instead of adding an isolated vertex
  $z$ with weight $n+m-\binom{k}{2}-k$ in the last step, we add a
  clique of size $n+m-\binom{k}{2}-k$.  We set $p = n+m-\binom{k}{2}$
  as before.  Using Lemma~\ref{l-incidence} and arguments similar to
  the ones in the proof of Theorem~\ref{t-splithard}, it is not hard
  to show that the obtained instance $(G'',p)$ of \pname{Vertex
    Integrity} is a yes-instance if and only if $(G,k)$ is a
  yes-instance of \pname{Clique}.  The observation that $G''$ is a
  chordal graph completes the proof.
\end{proof}

Recall that Fellows and Stueckle~\cite{fellows1989immersion} proved
that \vi{} can be solved in time $O(p^{3p} n)$.  Their arguments can
be slightly strengthened to yield the following result.

\begin{theorem}
\label{t-vi-branching}
\wvi{} can be solved in $O(p^{p+1} n)$ time.
\end{theorem}
\begin{proof}
  Let $(G,w,p)$ be an instance of \wvi, and let $n = |V(G)|$ and $m =
  |E(G)|$.  We assume that every vertex in $G$ has weight at
  least~$1$, as vertices of weight~$0$ can simply be deleted from the
  graph.  This implies in particular that $|X|\leq w(X)$ for every set
  $X\subseteq V(G)$.  We now show that we may also assume that $m\leq
  (p-1)n$.  Suppose that $(G,w,p)$ is a yes-instance.  Then there
  exists a set $X\subseteq V(G)$ such that $w(X)+\wcc(G-X) \leq p$.
  Let $G_1,\ldots,G_r$ be the components of $G-X$.  Since every vertex
  has weight at least~$1$, it holds that $|X \cup V(G_i)| \leq w(X
  \cup V(G_i)) \leq p$ for each $i \in \{1, \ldots, r\}$.  Observe
  that $G$ has a path decomposition of width at most~$p-1$ whose bags
  are exactly the sets $X \cup V(G_i)$.  This implies that the
  pathwidth, and hence the treewidth, of $G$ is at most~$p-1$.  It is
  well-known that every $n$-vertex graph of treewidth at most~$t$ has
  at most~$tn$ edges~\cite{BodlaenderF05}.  We thus conclude that if
  $(G,w,p)$ is a yes-instance, then $m\leq (p-1)n$.  Our algorithm can
  therefore safely reject the instance if $m>(p-1)n$.

  We now describe a simple branching algorithm that solves the
  problem.  At each step of the algorithm, we use a depth-first search
  to find a set $U$ of at most~$p+1$ vertices such that $G[U]$ is
  connected and $w(U) \geq p+1$.  If such a set does not exist, then
  every component of the graph under consideration has weight at
  most~$p$, so the empty set is an $\iota$-set of the graph and we are
  done.  Otherwise, we know that any $\iota$-set of the graph contains
  a vertex of~$U$.  We therefore branch into $|U| \leq p+1$
  subproblems: for every $v \in U$, we create the instance $(G-v, w,
  p-w(v))$, where we discard the instance in case $p - w(v) < 0$.
  Since the parameter $p$ decreases by at least~$1$ at each branching
  step, the corresponding search tree~$T$ has depth at most~$p$.
  Since~$T$ is a $p+1$-ary tree, it contains $O(p^p)$ nodes in total.
  Due to the assumption that $m\leq (p-1)n$, the depth-first search at
  each step can be performed in time $O(pn)$.  This yields an overall
  running time of $O(p^p pn) = O(p^{p+1} n)$.
\end{proof}

We prove that the problem admits a polynomial kernel with respect to
parameter~$p$.

\begin{theorem}
  \label{t-vikernel}
  \wvi{} admits a kernel with at most~$p^3$ vertices, where each
  vertex has weight at most~$p$.
\end{theorem}
\begin{proof}
  We describe a kernelization algorithm for the problem.  Let
  $(G,w,p)$ be an instance of \wvi{}.  We first delete all vertices of
  weight~$0$ without changing the parameter.  Observe that after this
  preprocessing step, the weight of every vertex is at least~$1$, and
  hence $|X|\leq w(X)$ for every set $X\subseteq V(G)$.  We apply the
  following two reduction rules.

  Our first reduction rule starts by sorting the components of $G$
  according to their weights.  Let $G_1,\ldots,G_r$ be the components
  of $G$ such that $\wcc(G_1) \geq \wcc(G_2) \geq \cdots \geq
  \wcc(G_r)$.  If $r > p+1$, then we delete the component $G_i$ for
  every $i \in \{p+2,\ldots,r\}$, without changing the parameter.  In
  other words, we keep only the $p+1$ heaviest components of $G$.  Let
  $G'$ be the obtained graph.  To see why this rule is safe, it
  suffices to prove that $(G,w,p)$ is a yes-instance if the new
  instance $(G',w,p)$ is a yes-instance, as the reverse direction
  trivially holds.  Suppose $(G',w,p)$ is a yes-instance.  Then there
  exists a set $X\subseteq V(G')$ such that $w(X)+\wcc(G'-X)\leq p$.
  Since $|X|\leq w(X)\leq p$ and $G'$ has exactly $p+1$ components,
  there exists an index $i\in \{1,\ldots,p+1\}$ such that $X$ does not
  contain any vertex from $G_i$.  Since $\wcc(G_i) \geq \wcc(G_j)$ for
  every $j \in \{p+2,\ldots,r\}$, it holds that $\wcc(G - X) = \wcc(G'
  - X)$.  Hence $w(X) + \wcc(G - X) = w(X)+\wcc(G'-X)\leq p$, implying
  that $\iota(G)\leq p$ and that $(G,w,p)$ is a yes-instance.

  The second reduction rule checks whether there exists a vertex~$v
  \in V(G)$ for which $w(N_G[v]) > p$.  Suppose such a vertex~$v$
  exists.  If $p-w(v)\geq 0$, then we delete~$v$ from the graph and
  reduce the parameter~$p$ by~$w(v)$.  If $p-w(v)<0$, then we return a
  trivial no-instance.  To see why this is safe, it suffices to show
  that if $(G,w,p)$ is a yes-instance, then~$v$ belongs to any
  $\iota$-set of~$G$.  Suppose $(G,w,p)$ indeed is a yes-instance, and
  let~$X$ be an $\iota$-set of~$G$.  Then $w(X)+\wcc(G-X) =
  \iota(G)\leq p$.  For contradiction, suppose that $v\notin X$.
  Consider the component~$H$ of~$G-X$ that contains~$v$.  Since every
  vertex of~$N_G[v]$ belongs either to~$X$ or to component~$H$, and
  $w(N_G[v]) > p$ by assumption, we find that $w(X) + \wcc(H) > p$.
  But this implies that $w(X) + \wcc(G-X) \geq w(X) + \wcc(H) > p$,
  yielding the desired contradiction.

  Let $(G',w,p')$ denote the instance obtained after exhaustively
  applying the above reduction rules, where~$w$ denotes the
  restriction of the original weight function to the vertices of~$G'$.
  Observe that~$p'\leq p$.  We assume that $p'\geq 2$, as otherwise we
  can trivially solve the instance $(G',w,p')$.  We claim that if
  $(G',w,p')$ is a yes-instance, then $|V(G')|\leq p^3$.  Suppose
  $(G',w,p')$ is a yes-instance, and let $X\subseteq V(G')$ be an
  $\iota$-set of $G'$.  Then $w(X) + \wcc(H) \leq p' \leq p$ for every
  component $H$ of $G'-X$.  This, together with the fact that every
  vertex in $G'$ has weight at least~$1$, implies that $|X|\leq p$ and
  $|H| \leq \wcc(H) \leq p-w(X) \leq p-|X|$ for every component $H$ of $G'-X$.
  Since the first reduction rule cannot be applied on the instance
  $(G',w,p')$, we know that $G'$ has at most $p'+1\leq p+1$
  components.  If $X = \emptyset$, then each of these components
  contains at most $p$ vertices, so $|V(G')| \leq (p+1) p \leq p^3$,
  where the last inequality follows from the assumption that $p\geq
  p'\geq 2$.  Now suppose $|X|\geq 1$.  Observe that every vertex
  in~$X$ has degree at most~$p'\leq p$ due to the assumption that the
  second reduction rule cannot be applied.  Hence, every vertex of $X$
  is adjacent to at most $p$ components of $G'-X$, implying that there
  are at most $p^2$ components of $G'-X$ that are adjacent to $X$.
  Since $G'$ itself has at most $p+1$ components, at least one of
  which contains a vertex of $X$, we find that $G'-X$ has at most
  $p^2+p$ components in total.  Recall that each of these components
  contains at most~$p-|X|$ vertices.  We conclude that $|V(G')|\leq
  (p^2+p)(p-|X|) + |X| \leq p^3$, where we use the assumption that
  $|X|\geq 1$.  Due to the second reduction rule, each vertex in $G'$
  has weight at most~$p$.

  It remains to argue that our kernelization algorithm runs in
  polynomial time.  Observe that the execution of any reduction rule
  strictly decreases either the number of vertices in the graph or the
  parameter, so each rule is applied only a polynomial number of
  times.  The observation that each rule can be executed in polynomial
  time completes the proof.
\end{proof}

\section{Component Order Connectivity}
\label{s-coc}

It is easy to see that $(G,p)$ is a yes-instance of \vi{} if and only
if there exist non-negative integers $k$ and $\ell$ with $k+\ell = p$
such that $(G,k,\ell)$ is a yes-instance of \coc.  Hence, any instance
$(G,p)$ of \vi{} can be solved by making at most~$p$ calls to an
algorithm solving \coc{},
implying that \coc{} cannot be solved in polynomial time on any graph
class for which \vi{} is \NP-complete, unless~\P=\NP.

Our next two results identify graph classes for which wCOC and COC are
strictly harder than wVI and VI, respectively.

\begin{theorem}
  \label{t-complete}
  \pname{Weighted Component Order Connectivity} is weakly
  \NP-complete on complete graphs.
\end{theorem}
\begin{proof}
  We reduce from \pname{Partition}, which is the problem of
  determining whether a multiset $A$ of positive integers can be
  partitioned into two subsets $A_1$ and $A_2$ such that the sum of
  the elements in $A_1$ equals the sum of the elements in $A_2$.  This
  problem is well-known to be \NP-complete~\cite{garey1979computers}.
  Given an instance $A$ of \pname{Partition} with $n$ elements $(a_1,
  \dots, a_n)$, we construct an instance $(G,w,k,\ell)$ of \wcoc{} as
  follows.  We define $G$ to be a complete graph with vertex set $V =
  \{v_1,\ldots,v_n\}$, and the weight function $w$ is defined by
  setting $w(v_i) = a_i$ for every $i\in \{1,\ldots,n\}$.  Let $W =
  1/2 \sum_{i = 1}^n a_i$.  We set $k = \ell = W$.

  Suppose~$A$ can be partitioned into two subsets~$A_1$ and~$A_2$ such
  that the sum of the elements in~$A_1$ equals the sum of the elements
  in~$A_2$ equals~$W$.  Let $X = \{v_i\in V \mid a_i\in A_1\}$ be the
  subset of vertices of~$G$ corresponding to the set~$A_1$.  Since~$G$
  is a complete graph and $w(V) = 2W$, is clear that $w(X) = W$ and
  $\wcc(G-X) = W$, implying that $(G,w,k,\ell)$ is a yes-instance of
  \wcoc.  The reverse direction is similar: if there exists a subset
  $X'\subseteq V$ with $w(X') = W$ and $\wcc(G-X') = W$, then the
  partition $X',V\setminus X'$ of~$V$ corresponds to a desired
  partition $A_1,A_2$ of~$A$.
\end{proof}

\begin{theorem}
  \label{t-splitcoc}
  \coc{} is \NP-complete on split graphs.
\end{theorem}
\begin{proof}
  We give a reduction from the \NP-hard problem \pname{Clique}.  Let
  $(G,k)$ be an instance of \pname{Clique} with $n = |V(G)|$ and $m =
  |E(G)|$.  Let $G^* = (C^*,I^*,E^*)$ be the split incidence graph of
  $G$, and let $\ell = n+m-\binom{k}{2}$.  By Lemma~\ref{l-incidence},
  there is a clique of size~$k$ in $G$ if and only if there exists a
  set $X\subseteq C^*$ such that $|X|\leq k$ and $n(G^*-X)\leq
  n+m-\binom{k}{2}$.  This immediately implies that $(G,k)$ is a
  yes-instance of \pname{Clique} if and only if $(G^*,k,\ell)$ is a
  yes-instance of \coc.
\end{proof}

We now present a pseudo-polynomial time algorithm, called \algWInt,
that solves \wcoc{} in $O(kn^3)$ time on interval graphs.  We refer to
Figure~\ref{alg:weighted-interval-alg} for pseudocode of the
algorithm.
Finally, we show that we can easily modify the algorithm to also run
in time $O( \ell n^3 )$.

\begin{figure}[ht!]
\begin{algorithm}[H]
\smallskip
Algorithm \algWInt\\
\smallskip
  \KwIn{An instance $(G,w,k,\ell)$ of \wcoc, where $G$ is an interval graph}
  \KwOut{``yes'' if $(G,w,k,\ell)$ is a yes-instance, and ``no'' otherwise}
\BlankLine
\BlankLine
Remove every vertex of weight $0$ from $G$\\
Construct $K_0, \dots, K_{t+1}$ \\
Construct $S_0, \dots, S_t$ \\
Construct $V_{i,j}$ for every $0\leq i<j\leq t$
\BlankLine
\BlankLine
Set all elements of $\dpt$ to $k+1$ \\
Set $\dpt[0] = 0$\\
\BlankLine

\For{$j$ \textup{from} $1$ \textup{to} $t$} {
    \For{$i$ \textup{from} $j-1$ \textup{to} $0$} {
        Let $v_1, \dots, v_{|V_{i,j}|}$ be the vertices of $V_{i,j}$ \\
        Let $w_p=w(v_p)$ for every $p\in \{1,\ldots,|V_{i,j}|\}$ \\
        \If{$w(V_{i,j}) \leq k+\ell$} {
        Let $I = \algMinSup( (w_1, \dots, w_{|V_{i,j}|}),w(V_{i,j}) - \ell)$ \\
        Let $Y_{i,j}=\{v_p \in V_{i,j} \mid p\in I\}$ \\
        $\dpt[j] = \min \begin{cases}
            \dpt[j] \\
            \dpt[i] + w(Y_{i,j}) + w(S_j\setminus S_i)
        \end{cases}$
    }
    }
}

\BlankLine

\KwRet ``yes'' if $\dpt[t]\leq k$, and ``no'' otherwise
\vspace{.3cm}
\end{algorithm}
\caption{Pseudocode of the algorithm $\algWInt$ that solves the \wcoc{}
  problem on interval graphs in $O(kn^3)$ time.}
  \label{alg:weighted-interval-alg}
\end{figure}

Given an instance~$(G,w,k,\ell)$, where~$G$ is an interval graph, the
algorithm first removes every vertex of weight~$0$.  It then computes
a \emph{clique path} of~$G$, i.e., an ordering $K_1,\ldots,K_t$ of the
maximal cliques of~$G$ such that for every vertex $v\in V(G)$, the
maximal cliques containing~$v$ appear consecutively in this ordering.
Since~$G$ is an interval graph, such an ordering exists and can be
obtained in~$O(n^2)$ time~\cite{BoothL76}.  For convenience, we define
two empty sets~$K_0$ and~$K_{t+1}$.  The algorithm now computes the
set $S_i = K_i \cap K_{i+1}$ for every $i\in \{0,\ldots,t\}$.  Observe
that~$S_0$ and~$S_t$ are both empty by construction, and that the
non-empty sets among $S_1,\dots, S_{t-1}$ are exactly the minimal
separators of $G$ (see, e.g.,~\cite{ho1989counting}).  For every $q
\in \{0,\ldots,t+1\}$, we define $G_q = G[\bigcup_{i = 0}^q K_i]$.
Also, for any two integers $i,j$ with $0\leq i<j\leq t$, the algorithm
computes the set
\[
V_{i,j} = \bigcup_{p = i+1}^j K_p\setminus (S_i\cup S_j).
\]
Informally speaking, the set~$V_{i,j}$ consists of the vertices of~$G$
that lie ``in between'' separators~$S_i$ and~$S_j$.

\medskip

Let us give some intuition behind the next phase of the algorithm.
Suppose $(G,w,k,\ell)$ is a yes-instance of \wcoc{}, and let~$X$ be a
solution for this instance.  Generally speaking,~$X$ fully contains
some minimal separators of~$G$ whose removal is necessary to break the
graph into pieces, as well as additional vertices that are deleted
from these pieces with the sole purpose of decreasing the weight of
each piece to at most~$\ell$.  The constructed clique path
$K_1,\ldots,K_t$ corresponds to a linear order of the minimal
separators $S_1,\ldots,S_{t-1}$ of~$G$.  We will use this linear
structure to find a minimum solution by doing dynamic programming over
the minimal separators of~$G$.

For every $q\in \{0,\ldots,t\}$, let~$k_q$ denote the smallest integer
such that there exists a set $X \subseteq V(G)$ satisfying the
following three properties:
\begin{itemize}
\item $w(X) = k_q$;
\item $S_q$ is a subset of $X$;
\item $X$ is a solution for the instance $(G_q,w, k_q, \ell)$.
\end{itemize}
In other words,~$X$ is a ``cheapest'' solution for $(G_q,w,k_q,\ell)$
that fully contains the minimal separator~$S_q$.
The algorithm now constructs an array $\dpt$ with $t+1$ entries, each
of which is an integer from $\{0,\ldots,k+1\}$.  Initially, all the
elements of the array are set to $k+1$.  For any $q\in
\{0,\ldots,t\}$, we say that the entry $\dpt[q]$ has \emph{reached
  optimality} if
\[
\dpt[q] = \begin{cases}
  k_q \mbox{\indent\indent if $k_q \leq k$} \\
  k+1 \mbox{\indent otherwise\,.}
\end{cases}
\]     
Recall that $S_t = \emptyset$ and that $G_t = G$.  Hence, if $\dpt[t]$ has
reached optimality, then the input instance $(G,w,k,\ell)$ is a
yes-instance if and only if $\dpt[t]\leq k$.

The algorithm uses a subroutine $\algMinSup$ that, given a multiset
of~$r$ weights $(w_1, \ldots, w_r)$ and a target~$W$ such that
$\sum_{i = 1}^r w_i \geq W$, finds a set $I \subseteq \{1, \dots, r\}$
such that $\sum_{i \in I} w_i$ is minimized with respect to the
constraint $\sum_{i \in I} w_i \geq W$.  Note that this subroutine
$\algMinSup$ can be implemented to run in time~$O(Wr)$ using the
classical dynamic programming algorithm for \pname{Subset Sum}.

\begin{lemma}
  \label{lemma:algwint-correctness}
  Given an instance $(G,w,k,\ell)$ of \wcoc, where $G$ is an interval
  graph, the algorithm $\algWInt$ returns ``yes'' if and only if
  $(G,w,k,\ell)$ is a yes-instance, and ``no'' otherwise.
\end{lemma}

\begin{proof}
  Recall that in order to prove the lemma, it suffices to prove that
  by the end of the algorithm, $\dpt[t]$ has reached optimality.  For
  each $j\in \{1,\ldots, t\}$, we define $P_1(j)$ to be the statement
  ``at the start of iteration $j$ of the outer loop, $\dpt[i]$ has
  reached optimality for every $i<j$,'' and we define $P_2(j)$ to be
  the statement ``at the end of iteration $j$ of the outer loop,
  $\dpt[i]$ has reached optimality for every $i\leq j$.'' 

\smallskip
\noindent
\emph{Claim 1.  For any $q\in \{1,\ldots,t-1\}$, it holds that $P_2(q)$
  implies $P_1(q+1)$.}

 \medskip
Observe that Claim~1 trivially holds. We also need the following claim.

 \medskip
 \noindent
 \emph{Claim 2.  For any $q\in \{1,\ldots,t\}$, it holds that $P_1(q)$
   implies $P_2(q)$.}

\smallskip
\noindent
In order to prove Claim~2, we first prove that if $\dpt[q] \leq k$ at the
end of iteration~$q$, then there is a solution of weight $\dpt[q]$ for
the instance $(G_q,w,\dpt[q],\ell)$ that contains $S_q$.  Let $r < q$
be such that $\dpt[q] = \dpt[r] + w(Y_{r,q}) + w(S_q \setminus S_r)$.
Due to our initialization of the table~$\dpt$ and the assumption that
$\dpt[q] \leq k$, such an~$r$ exists.  Because we assume that $P_1(q)$
holds, $\dpt[r]$ has reached optimality.  Hence, there is a set
$X_r\subseteq V(G_r)$ such that $w(X_r) = \dpt[r]$,~$X_r$
contains~$S_r$, and~$X_r$ is a solution for $(G_r,w,\dpt[r],\ell)$.
Consider the set~$Y_{r,q}$ that was constructed using the subroutine
$\algMinSup$.  Recall that $w(Y_{r,q})\geq w(V_{r,q})-\ell$ and hence
$w(V_{r,q}) \leq w(Y_{r,q}) + \ell$.  Furthermore, by assumption, the
only component of $G_q - (X_r \cup S_q)$ possibly of weight larger
than~$\ell$ is a component of $G[V_{r,q}]$.  Let $X_q = X_r \cup S_q
\cup Y_{r,q}$.  Due to the correctness of \algMinSup, it holds that
any component of~$G_q - X_q$ is of weight at most~$\ell$.  Hence~$X_q$
is a solution for $(G_q,w,\dpt[q],\ell)$ that contains~$S_q$ and has
weight~$\dpt[q]$.

It remains to prove that for any set $X\subseteq V(G_q)$ such that $S_q
\subseteq X$ and $\wcc(G_q-X)\leq \ell$, it holds that $w(X) \geq
\dpt[q]$.  Assume, for contradiction, that there exists a set
$X\subseteq V(G_q)$ such that $S_q \subseteq X$, $\wcc(G_q-X)\leq
\ell$, and $w(X) < \dpt[q]$.  Let $r < q$ be the largest index such
that $S_r \subseteq X$.  Observe that $r$ exists due to the fact that
$S_0 = \emptyset$.  We claim that
\[
w(V_{r,q}\setminus X)>\ell \, .
\]
First consider the case where $w(V_{r,q})>k+\ell$.  Then
$w(V_{r,q}\setminus X)\geq w(V_{r,q})-w(X)>k+\ell-w(X)\geq \ell$,
where the last inequality follows from the fact that $w(X)<\dpt[q]\leq
k+1$.

Now consider the case where $w(V_{r,q})\leq k+\ell$.  Then by the
algorithm, $\dpt[q]\leq \dpt[r]+w(Y_{r,q})+w(S_q\setminus S_r)$.  By
definition, we know that the sets $V(G_r)$, $V_{r,q}$, and
$S_q\setminus S_r$ form a partition of $V(G_q)$, and hence $w(X \cap
V(G_r)) + w(X \cap V_{r,q}) + w(X \cap (S_q \setminus S_r)) = w(X) <
\dpt[q] \leq \dpt[r] + w(Y_{r,q}) + w(S_q\setminus S_r)$.  Since $S_r
\subseteq X$ and $P_1(q)$ holds by assumption, we have that
$\dpt[r]\leq w(X \cap V(G_r))$.  From $S_q \subseteq X$ it follows
that $w(X \cap (S_q\setminus S_r)) = w(S_q\setminus S_r)$ and since we
know that $w(X \cap V(G_r)) \geq \dpt[r]$ it follows immediately that
$w(X \cap V_{r,q}) < w(Y_{r,q})$.  By the correctness of $\algMinSup$
it follows that $w(X \cap V_{r,q}) < w(V_{r,q}) - \ell$.  Hence, we
find that $w(V_{r,q} \setminus X) = w(V_{r,q}) - w(X \cap V_{r,q}) >
w(V_{r,q}) - (w(V_{r,q}) - \ell) = \ell$.
        
We now prove that~$G[V_{r,q} \setminus X]$ is connected.  Let~$u$ and~$v$ be
two distinct vertices from $V_{r,q} \setminus X$.  We will now prove
that~$u$ and~$v$ belong to the same component of $G[V_{r,q}\setminus
X]$.  If there is a maximal clique~$K_i$ of~$G$ containing both~$u$
and~$v$, then this trivially holds.  Suppose this is not the case.
Let~$a$ be the largest number such that~$u \in K_a$ and~$b$ the
smallest number such that~$v \in K_b$.  Assume without loss of
generality that $a < b$.  By construction of $V_{r,q}$, it follows
that $r < a < b \leq q$.  By the definition of $r$, we know that $S_z
\setminus X$ is non-empty for every $z \in [a,b)$.  Consequently, the
vertices in $\bigcup_{z \in [a,b)} S_z$ induce a connected subgraph of
$G[V_{r,q}\setminus X]$, and~$u$ and~$v$ are thus contained in the same
component of this graph.

Recall that $\wcc(G_q-X)\leq \ell$ by the definition of~$X$.  However,
$G[V_{r,q} \setminus X]$ is a component of~$G_q-X$ that has weight
$w(V_{r,q} \setminus X)>\ell$.  This yields the desired contradiction,
and completes the proof of Claim~2.
    
\medskip We now show that $P_1(1)$ holds.  In order to show this, it
suffices to argue that~$\dpt[0]$ has reached optimality at the start
of the first iteration of the outer loop.  Recall that~$\dpt[0]$ is
set to~$0$ during the initialization phase, so~$\dpt[0]$ equals~$0$ at
the start of the first iteration of the outer loop.  Since $S_0 =
V(G_0) = \emptyset$, it holds that $k_0 = 0$.

Since $P_1(1)$ holds, we can repeatedly apply Claims~1 and~2 to deduce
that $P_2(t)$ holds.  Hence, $\dpt[t]$ has reached optimality by the
end of the last iteration of the outer loop.
\end{proof}

\begin{theorem}
  \label{t-intervalweighted}
  \wcoc{} can be solved in $\otilde(\min\{k,\ell\}\cdot n^3)$ time on
  interval graphs.
\end{theorem}
\begin{proof}
  Due to Lemma~\ref{lemma:algwint-correctness}, it suffices to prove
  that the algorithm \algWInt{} runs in time $O(sn^3)$, where $s =
  \min\{k,\ell\}$.  Clearly, we can remove all vertices of weight $0$
  in $O(n^2)$ time.  It is well-known that a clique path of an
  interval graph can be constructed in $O(n^2)$ time, and that an
  interval graph has no more than~$n$ maximal cliques~\cite{Ibarra09}.
  Consequently, all the sets $K_0,\ldots,K_{t+1}$ and $S_0,\ldots,S_t$
  can be constructed in $O(n^2)$ time.  Observe that for all $0 \leq i
  < j \leq t$, it holds that $V_{i,j} = V_{i,j-1} \cup (V(K_j)
  \setminus (S_i \cup S_j))$.  Hence, once the sets
  $K_0,\ldots,K_{t+1}$ and $S_0,\ldots,S_t$ have been constructed, the
  sets $V_{i,j}$ can be computed in $O(n^3)$ time using a
  straightforward dynamic programming procedure.
  
  We claim that the body of the inner loop runs in time
  $\otilde\left(k n\right)$.  Observe that the body of this loop is
  only executed if $w(V_{i,j})\leq k+\ell$.  Since $|V_{i,j}|\leq n$
  and $w(V_{i,j})-\ell\leq k+\ell-\ell = k$, the algorithm
  \algMinSup{} solves the instance $( \{w_1, \dots,
  w_{|V_{i,j}|}\},w(V_{i,j}) - \ell)$ in $O(kn)$ time, which is
  therefore also the time it takes to obtain $Y_{i,j}$.  Clearly, the
  value of $\dpt[j]$ can be computed in $\otilde(n)$ time.  Since the
  inner loop is executed $O(n^2)$ times, we conclude that $\algWInt$
  terminates in time $\otilde(kn^3)$.
  
  It remains to argue why \wcoc{} can be solved in time $O(\ell n^3)$
  in case $\ell< k$.  Recall the following two lines from the inner
  loop of the algorithm \algWInt{}, explaining how we obtain the set
  $Y_{i,j}$:
  
  \smallskip
  \noindent {\small \hspace*{2.05cm} Let $I = \algMinSup( \{w_1, \dots,
    w_{|V_{i,j}|}\},w(V_{i,j}) - \ell)$\\
    \hspace*{2.05cm} Let $Y_{i,j} = \{v_p \in V_{i,j} \mid p\in I\}$ }
  
  \medskip
  \noindent
  The idea is to replace the subroutine $\algMinSup$ by a subroutine
  \algMaxInf{} that, given a multiset of weights $(w_1, \ldots, w_n)$
  and a target $W$, finds a set $I \subseteq \{1,\dots, n\}$ such that
  $\sum_{i \in I} w_i$ is maximized under the constraint $\sum_{i \in
    I} w_i \leq W$.  It is clear that \algMaxInf{}, just like
  \algMinSup{}, can be solved in $O(Wn)$ time.  By replacing the above
  two lines in the inner loop by the following two lines, we can
  obtain the exact same set $Y_{i,j}$ in $O(\ell n)$ time:
  
  \smallskip
  \noindent {\small
    \hspace*{2.05cm} Let $I  =  \algMaxInf((w_1, \dots, w_{|V_{i,j}|}),\ell)$ \\
    \hspace*{2.05cm} Let $Y_{i,j} = \{v_p \in V_{i,j} \mid p\notin I\}$ }
  
  \medskip
  \noindent
  This slight modification yields an algorithm for solving \wcoc{} on
  interval graphs in $\otilde(\ell n^3)$ time.
\end{proof}

\begin{theorem}
  \label{theorem:quad-time-coc}
  \coc{} can be solved in $O(n^2)$ time on interval graphs.
\end{theorem}
\begin{proof}
  We describe a modification of the algorithm \algWInt{}, called
  \ucoc, that solves the unweighted \coc{} problem in $O(n^2)$ time on
  interval graphs.  There are two reasons why the algorithm \algWInt{}
  does not run in $O(n^2)$ time: constructing all the sets $V_{i,j}$
  takes $O(n^3)$ time in total, and each of the $O(n^2)$ executions of
  the inner loop takes $O(kn)$ time, which is the time taken by the
  subroutine \algMinSup{} to compute the set $Y_{i,j}$ of vertices
  that are to be deleted.
  
  Recall that for every $j\in \{1,\ldots,t\}$ and every $i\in
  \{0,\ldots,j-1\}$, the set $Y_{i,j}$ computed by the algorithm
  \algWInt{} is defined to be the minimum-weight subset of $V_{i,j}$
  for which the weight of the subgraph $G[V_{i,j}]-Y_{i,j}$ is at most
  $\ell$.  Also recall that once the set $Y_{i,j}$ is computed, the
  value of $\dpt[j]$ is updated as follows:
  
  \medskip
  \noindent
  {\small
    \hspace*{2.05cm} $\dpt[j]  =  \min \begin{cases}
      \dpt[j] \\
      \dpt[i] + w(Y_{i,j}) + w(S_j\setminus S_i)
    \end{cases}$ }
  
  \bigskip When solving the unweighted variant of the problem, we can
  decrease the weight (i.e., order) of the subgraph $G[V_{i,j}]$ to at
  most $\ell$ by simply deleting $|V_{i,j}|-\ell$ vertices from
  $V_{i,j}$ in a greedy manner.  In other words, it is no longer
  important to decide \emph{which} vertices to delete from $V_{i,j}$,
  but only \emph{how many} vertices to delete.  This means that we can
  replace the entire body of the inner loop by the following line:
  
  \medskip
  \noindent
  {\small
    \hspace*{2.05cm} $\dpt[j]  =  \min \begin{cases}
      \dpt[j] \\
      \dpt[i] + (|V_{i,j}|-\ell) + |S_j\setminus S_i|
    \end{cases}$ }
  
  \medskip
  \noindent
  Hence
  it suffices to argue that we can precompute the values $|V_{i,j}|$
  and $|S_j\setminus S_i|$ for every $j\in \{1,\ldots,t\}$ and $i\in
  \{0,\ldots,j-1\}$ in $O(n^2)$ time in total.
  
  Recall that $V_{i,j} = \bigcup_{p = i+1}^j K_p\setminus (S_i \cup S_j)$ by
  definition, so
  \[
  |V_{i,j}| = |\bigcup_{p = i+1}^j K_p| - |S_i|-|S_j|+|S_i\cap S_j| \, .
  \]
  Moreover, it is clear that
  \[
  |S_j\setminus S_i| = |S_j| - |S_i\cap S_j| \, .
  \]
  
  The algorithm \ucoc{} starts by computing the sets
  $K_0,\ldots,K_{t+1}$ and $S_0,\ldots,S_t$ as before in $O(n^2)$
  time, as well as the cardinalities of these sets.  For each $v\in
  V(G)$, let $L(v)$ denote the largest index $i$ such that $v\in K_i$.
  Observe that we can compute the value $L(v)$ for all $v\in V(G)$ in
  $O(n^2)$ time in total.  The algorithm then computes the value
  $|\bigcup_{p = 0}^i K_p|$ for every $i \in \{0,\ldots,t\}$.  Using
  these values, it then computes the value
  \[
  |\bigcup_{p = i+1}^j K_p| = |\bigcup_{q = 0}^j K_q| - |\bigcup_{r =
    0}^i K_r|+|K_i\cap K_{i+1}| = |\bigcup_{q = 0}^j K_q| -
  |\bigcup_{r = 0}^i K_r|+|S_i|
  \]
  for every $j \in \{1, \ldots, t\}$ and every $i\in \{0, \ldots, j - 1
  \}$.  Observe that this can also be done in $O(n^2)$ time in total
  since all the terms in the expression have been precomputed.
  
  It remains to show that we can compute the value $|S_i\cap S_j|$ for
  all indices $i$ and $j$ with $0 \leq i<j\leq t$ in $O(n^2)$ time in
  total.  Let us fix an index $i\in \{0,\ldots,t\}$.  Since we
  precomputed the $L$-value of each vertex and we can order the
  vertices in $S_i$ by increasing $L$-value in $O(n)$ time, we can
  compute the value $|S_i \cap S_j| = |\{v\in S_i \mid L(v) \geq j+1\}|
  = |S_i \cap S_{j-1}| - |\{v \in S_i \mid L(v) = j+1 \}|$
  for all $j \in \{i+1,\ldots,t\}$.
  Observe that the expression $|\{v \in S_i \mid L(v) = j+1 \}|$ can be
  computed for every $j$ by one sweep through $S_i$ since $S_i$ is
  ordered by $L$-values.  Hence the computation of $|S_i \cap S_j|$,
  for a fixed $i$ and every $j$ can be performed in $O(n)$ time.
  This completes the proof.
\end{proof}

To conclude this section, we investigate the parameterized complexity
and kernelization complexity of COC and wCOC.  As mentioned in the
introduction, both problems are para-\NP-hard when parameterized
by~$\ell$ due to the fact that \coc{} is equivalent to \pname{Vertex
  Cover} when $\ell = 1$.  Our next result shows that when restricted
to split graphs, both problems are $\W[1]$-hard when parameterized
by~$k$ or by~$\ell$.

\begin{theorem}
  \label{t-whardsplit}
  \coc{} is $\W[1]$-hard on split graphs when parameterized by~$k$ or
  by~$\ell$.
\end{theorem}
\begin{proof}
  The fact that \coc{} is $\W[1]$-hard on split graphs when
  parameterized by~$k$ readily follows from the observation that the
  reduction in the proof of Theorem~\ref{t-splitcoc} is
  parameter-preserving and the fact that \textsc{Clique} is
  $\W[1]$-hard when parameterized by the size of the
  solution~\cite{ParameterizedComplexity}.
  
  To prove that the problem is $\W[1]$-hard on split graphs when
  parameterized by~$\ell$, we give a slightly different reduction from
  \textsc{Clique}.  Let $(G,q)$ be an instance of \textsc{Clique}, and
  construct $G^\dagger = (V^\dagger,E^\dagger)$, where $V^\dagger =
  \{v_x \mid x \in V(G) \} \cup \{w_e \mid e \in E(G)\}$ and
  $E^\dagger = \{w_{e_1}w_{e_2} \mid e_1, e_2 \in E(G) \} \cup
  \{v_xw_e \mid \text{vertex } x \text{ incident to edge } e \text{ in
  } G\}$.  Define $C^\dagger = \{v_e \mid e \in E(G)\}$ and $I^\dagger
  = V^\dagger \setminus C^\dagger$.  We also define $k= |E(G)| -
  \binom{q}{2}$ and $\ell = \binom{q}{2} + q$.  We will show that
  $(G,q)$ is a yes-instance of \textsc{Clique} if and only if
  $(G^\dagger,k,\ell)$ is a yes-instance of COC.
  
  First assume $(G,q)$ is a yes-instance of \textsc{Clique}, and let
  $Q \subseteq V(G)$ be a clique of size~$q$.  Define $Q^\dagger =
  \{w_e \mid e=uv \text{ for } u,v \in Q\}$.  Let $X^\dagger =
  C^\dagger \setminus Q^\dagger$ and consider $|X^\dagger|$ and
  $G^\dagger - X^\dagger$.  Observe that $|X^\dagger| = |C^\dagger
  \setminus Q^\dagger| = |C^\dagger| - |Q^\dagger| = |E(G)| -
  \binom{q}{2} = k$.  Also note that the neighborhood of~$Q^\dagger$
  in~$I^\dagger$ has size exactly~$q$. Hence the component of
  $G^\dagger - X^\dagger$ containing the vertices of $Q^\dagger$ has
  $|Q^\dagger| + q = \binom{q}{2} + q = \ell$ vertices, while every
  other component of $G^\dagger-X^\dagger$ contains exactly one
  vertex.  This implies that $(G^\dagger,k,\ell)$ is a yes-instance of
  COC.
  
  For the reverse direction, suppose that $(G^\dagger,k,\ell)$ is a
  yes-instance of COC.  Then there exists a set $X^\dagger \subseteq
  V^\dagger$ such that $|X^\dagger| \leq k$ and $n(G^\dagger -
  X^\dagger) \leq \ell$; let us call such a set $X^\dagger$ a {\em
    deletion set}.  Without loss of generality, assume that among all
  deletion sets, $X^\dagger$ contains the smallest number of vertices
  from $I^\dagger$.  We claim that $X^\dagger \cap I^\dagger=\emptyset$, i.e.,
  $X^\dagger\subseteq C^\dagger$.
  
  For contradiction, suppose there is a vertex $v\in X^\dagger \cap
  I^\dagger$.  If all the neighbors of $v$ belong to $X^\dagger$, then
  $X^\dagger\setminus \{v\}$ is a deletion set, contradicting the
  choice of $X^\dagger$.  Hence we may assume that there exists a
  vertex $w \in N_{G^\dagger}(v) \setminus X^\dagger$.  Let $D$ be the
  component of $G^\dagger - X^\dagger$ containing~$w$.  Observe that
  every component of $G^\dagger - X^\dagger$ other than $D$ has
  exactly one vertex, so $|V(D)|=n(G^\dagger - X^\dagger)$.  Let
  $X'=X^\dagger\setminus \{v\}$, and let $D'$ be the component of
  $G^\dagger - X'$ containing $v$ and $w$.  It is clear that
  $|V(D')|=|V(D)|+1$ and all components of $G^\dagger - X'$ other than
  $D'$ contain exactly one vertex.  Finally, let
  $X''=(X^\dagger\setminus \{v\}) \cup \{w\}$.  Then every component
  of $G^\dagger - X''$ has at most $|V(D')|-1\leq |V(D)|$ vertices,
  implying that $n(G^\dagger - X'')\leq n(G^\dagger - X^\dagger)$.
  Hence $X''$ is a deletion set, contradicting the choice of
  $X^\dagger$.  This contradiction proves that $X^\dagger\subseteq C^\dagger$.
  
  Observe that $|C^\dagger \setminus X^\dagger|=|C^{\dagger}| - |X^\dagger|
  \geq |E(G)| - k = \binom{q}{2}$.  Let $Q^{\dagger}$ be any subset of
  $C^\dagger \setminus X^\dagger$ of size~$\binom{q}{2}$.  Let $D$ be
  the component of $G^\dagger-X^\dagger$ containing $Q^\dagger$.
  Since $X^\dagger$ is a deletion set, $|V(D)|\leq
  \ell=\binom{q}{2}+q$.  This implies that $Q^\dagger$ has at most $q$
  neighbors in $I^\dagger$.  
  By construction of~$G^\dagger$, it holds that~$Q^\dagger$ has
  exactly~$q$ neighbors in~$I^\dagger$.  These~$q$ neighbors
  correspond to a clique of size~$q$ in~$G$.
\end{proof}

On the positive side, our next result shows that both problems become
fixed-parameter tractable when parameterized by $k+\ell$.

\begin{theorem}
  \label{t-branching}
  \wcoc{} can be solved in time $O(\ell^k (k+\ell)n) = 2^{O(k \log
    \ell)} n$.
\end{theorem}

\begin{proof}
  Let $(G,w,k,\ell)$ be an instance of \wcoc, and let $n = |V(G)|$ and
  $m = |E(G)|$.  We assume that every vertex in $G$ has weight at
  least~$1$, as vertices of weight~$0$ can simply be deleted from the
  graph.  Suppose that $(G,w,k,\ell)$ is a yes-instance.  Then there
  exists a set $X\subseteq V(G)$ such that $w(X)\leq k$ and
  $\wcc(G-X)\leq \ell$.  Let $G_1, \ldots, G_r$ be the components of
  $G-X$.  We can construct a path decomposition of $G$ by taking as
  bags the sets $X \cup V(G_i)$ for all $i\in \{1,\ldots,r\}$.  Since
  every vertex has weight at least~$1$, we know that each bag contains
  at most $k+\ell$ vertices, implying that $G$ has treewidth at
  most~$k+\ell-1$.  Consequently, $G$ has at most $(k+\ell-1)n$
  edges~\cite{BodlaenderF05}.  We may therefore assume that $m\leq
  (k+\ell-1)n$, as our algorithm can safely reject the instance
  otherwise.
  
  We now describe a simple branching algorithm that solves the
  problem.  Now, at each step of the algorithm, we use a depth-first
  search to find a set $U \subseteq V(G)$ of at most~$\ell+1$ vertices
  such that $\wcc(G[U]) \geq \ell+1$ and $G[U]$ induces a connected
  subgraph.  If such a set does not exist, then every component of the
  graph has weight at most~$\ell$, so we are done.  Otherwise, we know
  that any solution contains a vertex of~$U$.  We therefore branch
  into $|U| \leq \ell+1$ subproblems: for every $v \in U$, we create
  the instance $(G-v, w, k-w(v), \ell)$, where we discard the instance
  in case $k-w(v) < 0$.  Since the parameter~$k$ decreases by at
  least~$1$ at each branching step, the corresponding search tree~$T$
  has depth at most~$k$.  Since~$T$ is an $(\ell+1)$-ary tree of depth
  at most~$k$, it has at most $((\ell+1)^{k+1}-1)/((\ell+1)-1) =
  O(\ell^k)$ nodes.  Due to the assumption that $m\leq (k+\ell-1)n$,
  the depth-first search at each step can be performed in time $O(n+m)
  = O((k+\ell)n)$.  This yields an overall running time of $O(\ell^k
  (k+\ell)n) = 2^{O(k \log \ell)} n$.
\end{proof}

We now show that the branching algorithm in Theorem~\ref{t-branching}
is in some sense best possible.  In order to make this statement
concrete, we need to introduce some additional terminology.

For $k\geq 3$, let $s_k$ be the infimum of the set of all positive real
numbers $\delta$ for which there exists an algorithm that solves
$k$-SAT in time $O(2^{\delta n})$, where $n$ denotes the number of
variables in the input formula.  The Exponential Time Hypothesis (ETH)
states that $s_k>0$ for any $k\geq 3$~\cite{ImpagliazzoP99,
  impagliazzo2001which}.  In particular, this implies that there is no
$2^{o(n)}$-time algorithm for solving $3$-SAT, unless the ETH fails.
Lokshtanov, Marx, and Saurabh~\cite{LokshtanovMS11} developed a
framework for proving lower bounds on the running time of
parameterized algorithms for certain natural problems, assuming the
validity of the ETH.  In order to obtain these results, they proved
lower bounds for constrained variants of some basic problems such as
the following:

\medskip
\indent \pname{$k\times k$ Clique}\\
\indent \emph{Instance:} \hspace*{-.01cm} A graph $G$, and a partition
${\cal X}$ of $V(G)$ into $k$ sets $X_1,\ldots,X_k$ of size~$k$
each.\\
\indent \emph{Question:} Does $G$ have a clique $K$ such that $|K\cap
X_i| = 1$ for all $i\in \{1,\ldots,k\}$?

\begin{theorem}[\cite{LokshtanovMS11}]
  \label{t-lowerbound}
  There is no $2^{o(k \log k)}$ time algorithm for $k \times k$
  \pname{Clique}, unless the ETH fails.
\end{theorem}

Recall that the \wcoc{} problem can be solved in time $2^{O(k\log
  \ell)} n$ on general graphs.  We now show that the problem does not
admit a $2^{o(k\log \ell)} n^{O(1)}$-time algorithm, even when all the
vertices have unit weight and the input graph is a split graph, unless
the ETH fails.

\begin{theorem}
  \label{t-starspace}
  There is no $2^{o(k\log \ell)} n^{O(1)}$ time algorithm for \coc{},
  even when restricted to split graphs, unless the ETH fails.
\end{theorem}

\begin{proof}
  For contradiction, suppose there exists an algorithm $\mathbb{A}$
  for solving the \coc{} problem in time $2^{o(k \log \ell)}
  n^{O(1)}$.  Let $(G,{\cal X})$ be an instance of the $k\times k$
  \pname{Clique} problem, where ${\cal X} = \{X_1,\ldots,X_k\}$.  We
  assume that~$G$ contains no edge whose endpoints belong to the same
  set~$X_i$, as an equivalent instance can be obtained by deleting all
  such edges from $G$.  Due to this assumption, it holds that
  $(G,{\cal X})$ is a yes-instance of $k\times k$ \pname{Clique} if
  and only if~$G$ contains a clique of size~$k$.
  
  Now let $G^* = (C^*,I^*,E^*)$ be the incidence split graph of $G$,
  and let $\ell = |V(G)|+|E(G)|-\binom{k}{2}$.  By the definition of
  the $k\times k$ \pname{Clique} problem, we have that $|V(G)| = k^2$
  and $|E(G)|\leq k^2(k^2-1)/2$.  This implies that the graph $G^*$
  has at most $k^2+k^2(k^2-1)/2\leq k^4$ vertices, and that $\ell\leq
  k^4$.  By Lemma~\ref{l-incidence}, it holds that $(G^*,k,\ell)$ is a
  yes-instance of \coc{} if and only if $G$ has a clique of size~$k$.
  Hence, using algorithm~$\mathbb{A}$, we can decide in time $2^{o(k
    \log k^4)} k^{O(1)} = 2^{o(k \log k)}$ whether or not $(G,{\cal
    X})$ is a yes-instance of $k\times k$ \pname{Clique}, which by
  Theorem~\ref{t-lowerbound} is only possible if the ETH fails.
\end{proof}

We conclude this section by showing that the \wcoc{} problem admits a
polynomial kernel.  The arguments in the proof of
Theorem~\ref{t-coc-kernel} are similar to, but slightly different
from, those in the proof of Theorem~\ref{t-vikernel}.

\begin{theorem}
  \label{t-coc-kernel}
  \wcoc{} admits a kernel with at most $k\ell(k+\ell)+k$ vertices,
  where each vertex has weight at most~$k+\ell$.
\end{theorem}
\begin{proof}
  We describe a kernelization algorithm for the problem.  Let
  $(G,w,k,\ell)$ be an instance of \wcoc{}.  We first delete all
  vertices of weight~$0$ without changing the parameters.  Observe
  that after this first preprocessing step, the weight of every vertex
  is at least~$1$.  This implies in particular that $|X|\leq w(X)$ for
  every set $X\subseteq V(G)$.
  
  We now apply the following two reduction rules.  If $G$ contains a
  vertex $v$ such that $w(N_G[v])> k+\ell$, then we delete $v$ from
  $G$ and decrease $k$ by~$w(v)$, unless $w(v)>k$, in which case we
  output a trivial no-instance.  To see why this rule is safe, let us
  first show that $v$ belongs to any solution for the
  instance~$(G,w,k,\ell)$ if such a solution exists.  This follows
  from the observation that deleting any set $X\subseteq V(G)\setminus
  \{v\}$ with $w(X)\leq k$ from $G$ yields a graph $G'$ such that
  $w(N_{G'}[v])>\ell$.  For the same reason, there exists no solution
  if $w(v)>k$.  Our second reduction rule deletes any component $H$ of
  weight at most~$\ell$ from $G$ without changing either of the
  parameters.  This rule is safe due to the fact that $\wcc(H) \leq
  \ell$ implies that no minimum solution deletes any vertex from $H$.
  
  Let $(G',w,k',\ell)$ denote the instance that we obtain after
  exhaustively applying the above reduction rules, where $w$ denotes
  the restriction of the original weight function to the vertices of
  $G'$.  Observe that $k'\leq k$, while the parameter $\ell$ did not
  change in the kernelization process.  Suppose $X$ is a solution for
  this instance.  Then $w(X)\leq k'\leq k$, which implies that $X$
  contains at most~$k$ vertices.  For every component $H$ of $G'-X$,
  it holds that $|H| \leq \wcc(H) \leq \ell$, furthermore
  $H$ 
  is adjacent to at least one vertex of $X$, as otherwise our second
  reduction rule could have been applied.  Moreover, the fact that the
  first reduction rule cannot be applied guarantees that
  $w(N_G[v])\leq k+\ell$ for every $v\in V(G')$.  In particular, this
  implies that every vertex in $X$ has degree at most~$k+\ell$.  We
  find that $G-X$ has at most $k(k+\ell)$ components, each containing
  at most $\ell$ vertices.  We conclude that if $(G',w,k',\ell')$ is a
  yes-instance, then $|V(G')|\leq k\ell(k+\ell)+k$.  The observation
  that each vertex in $G'$ has weight at most~$k+\ell$ due to the
  first reduction rule completes the proof.
\end{proof}

\section{Concluding Remarks}


We showed that the \coc{} problem does not admit a $2^{o(k\log \ell)}
n^{O(1)}$ time algorithm, unless the ETH fails.  Can the problem be
solved in time $c^{k+\ell} n^{O(1)}$ for some constant $c$?
Similarly, it would be interesting to investigate whether it is
possible to solve \vi{} in time $c^p n^{O(1)}$ for some constant~$c$,
that is, does there exist a single-exponential time algorithm solving
\vi{}?

\section{Acknowledgements}
The authors are grateful to Daniel Lokshtanov for pointing out that
\pname{$k \times k$ Clique} admits no $2^{o(k \log k)}$ time algorithm
unless ETH fails.

\bibliographystyle{alpha}

\newcommand{\etalchar}[1]{$^{#1}$}

\end{document}